\documentclass[12pt]{article}%
\usepackage{algorithm}
\usepackage{algpseudocode}
\usepackage{alphabeta}
\usepackage{amsmath}
\usepackage{amsfonts}
\usepackage{amssymb}
\usepackage{boxedminipage}
\usepackage[open,openlevel=1]{bookmark}
\usepackage{caption}
\usepackage{color}
\usepackage{graphicx}
\usepackage{listings}
\usepackage[most]{tcolorbox}
\usepackage{multirow}%
\providecommand{\U}[1]{\protect \rule{.1in}{.1in}}
\newtheorem{theorem}{Theorem}[section]

\newtheorem{corollary}[theorem]{Corollary}

\newtheorem{proposition}[theorem]{Proposition}

\newenvironment{proof}[1][Proof]{\noindent \textbf{#1.} }{\  \rule{0.5em}{0.5em}}
\begin{document}

\title{A Game Theoretic Analysis\\of the Three-Gambler Ruin Game}
\author{Ath. Kehagias, G. Gkyzis, A. Karakoulakis, A. Kyprianidis}
\date{\today}
\maketitle

\begin{abstract}
We study the following game. Three players start with initial capitals of
$s_{1},s_{2},s_{3}$ dollars; in each round player $P_{m}$ is selected with
probability $\frac{1}{3}$; then \emph{he} selects player $P_{n}$ and they play
a game in which $P_{m}$ wins from (resp. loses to) $P_{n}$ one dollar with
probability $p_{mn}$ (resp. $p_{nm}=1-p_{mn}$). When a player loses all his
capital he drops out; the game continues until a single player wins by
collecting everybody's money.

This is a \textquotedblleft strategic\textquotedblright \ version of the
classical Gambler's Ruin game. It seems reasonable that a player may improve
his winning probability by judicious selection of which opponent to engage in
each round. We formulate the situation as a \emph{stochastic game} and prove
that it has at least one Nash equilibrium in deterministic stationary strategies.

\end{abstract}

\section{Introduction\label{sec01}}

In this paper we study a version of the three-gambler ruin problem in which,
whenever a gambler is \textquotedblleft activated\textquotedblright, he
\emph{chooses} his opponent. More specifically, we study our version is played
by the following rules.

\begin{enumerate}
\item The game is played in discrete time steps (rounds).

\item Three players start at the $0$-th round with initial capitals, of
$s_{1}$, $s_{2}$ and $s_{3}$ dollars.

\item In each round, one player is chosen with probability $\frac{1}{3}$ and
then \emph{he} chooses another player against whom he will play.

\item Suppose that the $m$-th player plays against the $n$-th one; with
probability $p_{mn}$ (resp. $p_{nm}=1-p_{mn\text{ }}$) he wins from (resp.
loses to) the $n$-th player one dollar.

\item The game continues until a single player wins (i.e. accumulates the
total $s_{1}+s_{2}+s_{3}$ dollars); if this never happens, the game continues
ad infinitum.
\end{enumerate}

\noindent The above game is a \textquotedblleft strategic\textquotedblright%
\ (in the \emph{game-theoretic} sense) version of the \textquotedblleft
classical\textquotedblright \ Gambler's Ruin (GR). Intuitively, it appears
reasonable that a player may improve his winning probability by judicious
opponent selection in each round. As far as we know, this game theoretic
approach to GR\ has not been explored previously.

The \textquotedblleft classic\textquotedblright \ Gambler's Ruin involves two
gamblers, is one of the earliest-studied probability problems (for a
historical review see \cite{Edwards1983,Song2013,Takacs1969}) and remains one
of the most popular introductory examples in the theory of Markov Chains. For
an overview of the basic results see \cite{Feller1991}.

An obvious generalization of GR\ is to have three or more gamblers. An early
attempt in this direction is \cite{Read1966} but the first (as far as we know)
major advance appeared in \cite{Sandell1988} and involved \textquotedblleft%
\emph{symmetric play}\textquotedblright, i.e., in each round all players have
equal probability of winning. Symmetric play was further studied in
\cite{Engel1993,Ross2009,Swan2006} and numerous other publications
\cite{Amano2001,Bruss2003,Chang1995,Cho1996,David2015,Grigorescu2016,Stirzaker1994,Stirzaker2006}%
. The case of \textquotedblleft \emph{asymmetric play}\textquotedblright, i.e.,
when players' winning probabilities are not necessarily equal, has also been
studied, first in \cite{Rocha1999,Rocha2004} and more recently in
\cite{Hashemi2011,Hussain2021,Hussain2023}. Finally versions played on graphs
have also been studied, for example in \cite{Neda2017}.

In all of the above works, no player \emph{strategy} is involved. In other
words, the players cannot influence who participates in a given round of the
game.\footnote{A notable exception appears in \cite{Ross2009}, where it is
stated that \textquotedblleft none of the preceding quantities [winning
probability, expected game duration etc.] depend on the rule for choosing the
players in each stage\textquotedblright. Here we have a case in which
strategies are implicitly considered, but they turn out to be irrelevant. This
is the case because the author studies \emph{symmetric play}; as we will see,
things are different for asymmetric play.}. Hence the evolution of the game is
governed by purely probabilistic laws.

As already mentioned, in this paper we take a different approach. In Section
\ref{sec02} we provide a rigorous formulation as a stochastic game. In Section
\ref{sec03} we provide an introductory analysis of the simple case in which
each player starts with initial capital of one dollar. In Section \ref{sec04}
we study the general case (i.e., with arbitrary initial capitals)\ and prove
that the game has at least one Nash equilibrium (NE)\ in deterministic
stationary strategies. \ In Section \ref{sec05} we discuss the computation of
the NE and provide several numerical experiments. In Section \ref{sec06} we
summarize our results and propose some future research directions.

It is worth emphasizing that our results can be generalized for the case of
more than three players; we limit ourselves to the three-player case mainly
for simplicity of presentation.

\section{The Game\label{sec02}}

We denote our game by $\Gamma_{3}\left(  \mathbf{p},K\right)  $. It involves
three players (gamblers) $P_{1},P_{2},P_{3}$ and has parameters $\mathbf{p}%
=\left(  p_{1},p_{2},p_{3}\right)  $ and $K$, which satisfy:%
\[
\forall n\in \left \{  1,2,3\right \}  :p_{n}\in \left(  0,1\right)  \text{ and
}K\in \left \{  3,4,,...\right \}  .
\]
$\Gamma_{3}\left(  \mathbf{p},K\right)  $ is played as follows:

\begin{enumerate}
\item At times $t\in \left \{  0,1,2,..\right \}  $ and for $n\in \left \{
1,2,3\right \}  $, $P_{n}$'s \emph{capital} is $s_{n}\left(  t\right)  $.

\item At $t=0$, the capitals satisfy $\  \sum_{n=1}^{3}s_{n}\left(  0\right)
=K$.

\item At $t\in \left \{  1,2,...\right \}  $ a player $P_{n}$ is selected
equiprobably from the set $\left \{  n:s_{n}\left(  t\right)  >0\right \}  $.

\item $P_{n}$ selects another player $P_{m}$ such that $s_{m}\left(  t\right)
>0$.

\item With probability $p_{nm}$ (resp. $p_{mn}$) $P_{n}$ receives one unit
from (resp. pays one unit to)$\ P_{m}$, where%
\[
p_{12}=p_{1},\text{\quad}p_{21}=1-p_{1},\quad p_{23}=p_{2},\text{\quad}%
p_{32}=1-p_{2},\quad p_{31}=p_{3},\text{\quad}p_{13}=1-p_{3}.
\]
Hence the players' capitals change as follows:

\begin{enumerate}
\item if $P_{n}$ wins then $s_{n}\left(  t\right)  =s_{n}\left(  t-1\right)
+1$ and $s_{m}\left(  t\right)  =s_{m}\left(  t-1\right)  -1$;

\item if $P_{n}$ loses then $s_{n}\left(  t\right)  =s_{n}\left(  t-1\right)
-1$ and $s_{m}\left(  t\right)  =s_{m}\left(  t-1\right)  +1$;

\item for $k\in \left \{  1,2,3\right \}  \backslash \left \{  n,m\right \}  $ we
have $s_{k}\left(  t\right)  =s_{k}\left(  t-1\right)  $.
\end{enumerate}

Obviously, at all $t$ we have $\sum_{n=1}^{3}s_{n}\left(  t\right)  =K$.

\item If at some time $t^{\prime}$ one player is left with zero capital, the
game continues between the two remaining players.

\item The game continues until at some time $t^{\prime \prime}$ there exists a
single player $P_{m}$ with $s_{m}\left(  t^{\prime \prime}\right)  =\sum
_{n=1}^{3}s_{n}\left(  0\right)  $, in which case this player is the winner.
If we such a player does not exist for any turn, the game continues ad infinitum.
\end{enumerate}

\bigskip

\noindent The \emph{game state} at time $t$ is $\mathbf{s}\left(  t\right)
=\left(  s_{1}\left(  t\right)  s_{2}\left(  t\right)  s_{3}\left(  t\right)
\right)  $. The \emph{state set} is
\[
S=\left \{  \left(  s_{1},s_{2},s_{3}\right)  :\forall n:s_{n}\geq0\text{ and
}\sum_{n=1}^{3}s_{n}=K\text{ }\right \}  .
\]
Note that each $\left(  s_{1},s_{2},s_{3}\right)  $ can be rewritten as
$\left(  s_{1},s_{2},K-s_{1}-s_{2}\right)  $. For $s_{1}=i$, we have $K-i+1$
possible states of the form $\left(  i,j,K-i-j\right)  $ with $j\in \left \{
0,1,...,K-i\right \}  $; since $i$ can take any value in $\left \{
0,1,...,K\right \}  $, the total number of states is
\[
N_{K}=\left \vert S\right \vert =\frac{\left(  K+1\right)  \left(  K+2\right)
}{2}.
\]
We define the following state sets:%
\[
S_{1}=\left \{  \left(  K,0,0\right)  \right \}  ,\quad S_{2}=\left \{  \left(
0,K,0\right)  \right \}  ,\quad S_{3}=\left \{  \left(  0,0,K\right)  \right \}
,\quad S_{i}=\left \{  \left(  s_{1},s_{2},s_{3}\right)  :\forall
n:s_{n}>0\right \}  .
\]
We will call the states $s\in S_{\tau}=S_{1}\cup S_{2}\cup S_{3}$
\emph{terminal}, the states $\mathbf{s}\in S_{i}$ \emph{interior} and the
states $s\in S_{b}=S\backslash S_{i}$ \emph{boundary}. It is easily checked
that $\left \vert S_{i}\right \vert =\frac{\left(  K-1\right)  \left(
K-2\right)  }{2}$ and $\left \vert S_{b}\right \vert =3K$. \noindent It will be
convenient to number the states so that: the first state is $(K,0,0)$, the
second is $(0,K,0)$ and the third is $(0,0,K)$; the remaining states can be
numbered arbitrarily.

\bigskip

\noindent The \emph{game history} at time $t$ is $\mathbf{h}\left(  t\right)
=\mathbf{s}\left(  0\right)  \mathbf{s}\left(  1\right)  ...\mathbf{s}\left(
t\right)  $. The set of all \emph{admissible} histories is denoted by $H$.
\emph{ A terminal history} is an $\mathbf{h}\left(  t\right)  =\mathbf{s}%
\left(  0\right)  \mathbf{s}\left(  1\right)  ...\mathbf{s}\left(  t\right)  $
such that $\mathbf{s}\left(  t\right)  \in S_{\tau}$.\emph{ }

\bigskip

\noindent For $n\in \left \{  1,2,3\right \}  $, $P_{n}$'s \emph{payoff} is
defined for every terminal history $\mathbf{h=s}\left(  0\right)
\mathbf{s}\left(  1\right)  ...\mathbf{s}\left(  t\right)  $ by
\[
\mathsf{Q}_{n}\left(  \mathbf{h}\right)  =\left \{
\begin{array}
[c]{ll}%
1 & \text{iff }\mathbf{s}_{n}\left(  t\right)  =K\\
0 & \text{else;}%
\end{array}
\right.
\]
payoff of nonterminal histories is zero.

\bigskip

\noindent A \emph{strategy} for $P_{n}$ is a function $\sigma_{n}%
:H\rightarrow \Delta_{3}\cup \left \{  \left(  0,0,0\right)  \right \}  $, where
\[
\Delta_{3}=\left \{  \left(  x_{n1},x_{n2},x_{n3}\right)  \text{ with }%
x_{nm}\geq0\text{ for }m\neq n\text{, }x_{nn}=0\text{, }\sum_{m=1}^{3}%
x_{nm}=1\right \}  .
\]
The interpretation is that%
\[
\sigma_{n}\left(  \mathbf{h}\right)  =\left(  x_{n1},x_{n2},x_{n3}\right)
\Leftrightarrow \left(  \forall m,n:\Pr \left(  P_{n}\text{ selects }P_{m}\text{
at }t\text{%
$\vert$%
}\mathbf{h}\right)  =x_{nm}\right)  .
\]
We will only consider \emph{admissible} strategies\footnote{For example, we
exclude strategies which assign positive probability to selecting a player
with zero capital, to a player selecting himself etc.}. \ A strategy profile
is a triple $\sigma=\left(  \sigma_{1},\sigma_{2},\sigma_{3}\right)  $. A
strategy $\sigma_{n}$ is called \emph{deterministic} iff
\[
\forall \mathbf{h}=\mathbf{s}\left(  0\right)  \mathbf{s}\left(  1\right)
...\mathbf{s}\left(  t\right)  :\sigma_{n}\left(  \mathbf{h}\right)  =\left(
x_{1}\left(  \mathbf{h}\right)  ,x_{2}\left(  \mathbf{h}\right)  ,x_{3}\left(
\mathbf{h}\right)  \right)  \text{ has a single nonzero element.}%
\]
A strategy $\sigma_{n}$ is called \emph{stationary} iff, for every
$\mathbf{h}$, $\sigma_{n}\left(  \mathbf{h}\right)  $ depends only on the last
state, i.e., iff%
\[
\forall \mathbf{h}=\mathbf{s}_{0}\mathbf{s}_{1}...\mathbf{s}_{t}:\sigma
_{n}\left(  \mathbf{h}\right)  =\sigma_{n}\left(  \mathbf{s}_{t}\right)  ;
\]
We will use the following shorter notation the $x_{mn}$'s:%
\[
x_{1}=x_{12}\text{ (and }1-x_{1}=x_{13}\text{),\quad}x_{2}=x_{23}\text{ (and
}1-x_{2}=x_{21}\text{),\quad}x_{3}=x_{31}\text{ (and }1-x_{3}=x_{32}\text{),}%
\]
and we will denote a stationary strategy profile by $\mathbf{x}=\left(
x_{1},x_{2},x_{3}\right)  $.

\bigskip

\noindent An initial state $\mathbf{s}\left(  0\right)  $ and a strategy
profile $\mathbf{\sigma}=\left(  \sigma_{1},\sigma_{2},\sigma_{3}\right)  $
define a probability measure on $H$. Hence the \emph{expected payoff} to
$P_{n}$\emph{ }is well defined by
\[
Q_{n}\left(  \mathbf{s}\left(  0\right)  ,\mathbf{\sigma}\right)
=\mathbb{E}_{\mathbf{s}\left(  0\right)  ,\mathbf{\sigma}}\left(
\mathsf{Q}_{n}\left(  \mathbf{h}\right)  \right)  .
\]
It is easily seen that $Q_{n}\left(  \mathbf{s}\left(  0\right)
,\mathbf{\sigma}\right)  $ is $P_{n}$'s probability of winning when the
starting state is $\mathbf{s}\left(  0\right)  $ and the players use the
strategy profile $\mathbf{\sigma}$.

\section{The case $K=3$\label{sec03}}

As a preliminary step in our analysis, let us consider the game when total
capital is $K=3$ (this is obviously the first nontrivial case)\ and, for
$n\in \left \{  1,2,3\right \}  $, $P_{n}$ uses the \emph{stationary} strategy
$\mathbf{x}_{n}=\left(  x_{n}\left(  \mathbf{s}\right)  \right)
_{\mathbf{s}\in S}$ . Suppressing the dependence on $\mathbf{x=}\left(
\mathbf{x}_{1},\mathbf{x}_{2},\mathbf{x}_{3}\right)  $, we denote $P_{n}$'s
payoff, when the game starts at state $\mathbf{s}$, by%
\[
V_{n}\left(  s\right)  =Q_{n}\left(  s,\mathbf{x}\right)  .
\]
Let us compute $\mathbf{V}_{1}\left(  \mathbf{s}\right)  =\left(  V_{1}\left(
\mathbf{s}\right)  \right)  _{\mathbf{s}\in S}$. First, we have
\[
V_{1}\left(  3,0,0\right)  =1,\quad V_{1}\left(  0,3,0\right)  =V_{1}\left(
0,0,3\right)  =V_{1}\left(  0,1,2\right)  =V_{1}\left(  0,2,1\right)  =0.
\]
Also, when two players are left in the game, the only admissible strategy for
each one is to select the other. For instance, in the state $\left(
1,2,0\right)  $, we must have
\[
x_{1}\left(  1,2,0\right)  =x_{12}\left(  1,2,0\right)  =1,\quad x_{2}\left(
1,2,0\right)  =x_{23}\left(  1,2,0\right)  =0.
\]
We have the following transition probabilites:%
\begin{align*}
\Pr \left(  \left(  2,1,0\right)  \rightarrow \left(  3,0,0\right)  \right)   &
=p_{1},\quad \Pr \left(  \left(  2,1,0\right)  \rightarrow \left(  1,2,0\right)
\right)  =1-p_{1},\\
\Pr \left(  \left(  1,2,0\right)  \rightarrow \left(  2,1,0\right)  \right)   &
=p_{1},\quad \Pr \left(  \left(  1,2,0\right)  \rightarrow \left(  0,3,0\right)
\right)  =1-p_{1}.
\end{align*}
Consequently%
\begin{align*}
V_{1}\left(  2,1,0\right)   &  =p_{1}V_{1}\left(  3,0,0\right)  +\left(
1-p_{1}\right)  V_{1}\left(  1,2,0\right)  ,\\
V_{1}\left(  1,2,0\right)   &  =p_{1}V_{1}\left(  2,1,0\right)  +\left(
1-p_{1}\right)  V_{1}\left(  0,3,0\right)  ,
\end{align*}
which becomes%
\begin{align*}
V_{1}\left(  2,1,0\right)   &  =p_{1}+\left(  1-p_{1}\right)  V_{1}\left(
1,2,0\right)  ,\\
V_{1}\left(  1,2,0\right)   &  =p_{1}V_{1}\left(  2,1,0\right)  .
\end{align*}
Solving the above system we get
\[
V_{1}\left(  2,1,0\right)  =\frac{p_{1}}{1-p_{1}+p_{1}^{2}},\quad V_{1}\left(
1,2,0\right)  =\frac{p_{1}^{2}}{1-p_{1}+p_{1}^{2}}.
\]
Similarly, we have
\begin{align*}
V_{1}\left(  2,0,1\right)   &  =\left(  1-p_{3}\right)  +p_{3}V_{1}\left(
1,0,2\right) \\
V_{1}\left(  1,0,2\right)   &  =\left(  1-p_{3}\right)  V_{1}\left(
2,0,1\right)
\end{align*}
and we get
\[
V_{1}\left(  2,0,1\right)  =\frac{1-p_{3}}{1-p_{3}+p_{3}^{2}},\quad
V_{1}\left(  1,0,2\right)  =\frac{1-2p_{3}+p_{3}^{2}}{1-p_{3}+p_{3}^{2}}.
\]
Finally, simplifying $x_{n}\left(  1,1,1\right)  $ to $x_{n}$ (for
$n\in \left \{  1,2,3\right \}  $) we have%
\begin{align*}
&  V_{1}\left(  1,1,1\right)  =\\
&  \hspace{5mm}\frac{1}{3}\left(  x_{1}\left(  p_{1}V_{1}\left(  2,0,1\right)
+\left(  1-p_{1}\right)  V_{1}\left(  0,2,1\right)  \right)  +\left(
1-x_{1}\right)  \left(  p_{3}V_{1}\left(  0,1,2\right)  +\left(
1-p_{3}\right)  V_{1}\left(  2,1,0\right)  \right)  \right) \\
&  +\frac{1}{3}\left(  x_{2}\left(  p_{2}V_{1}\left(  1,2,0\right)  +\left(
1-p_{2}\right)  V_{1}\left(  1,0,2\right)  \right)  +\left(  1-x_{2}\right)
\left(  p_{1}V_{1}\left(  2,0,1\right)  +\left(  1-p_{1}\right)  V_{1}\left(
0,2,1\right)  \right)  \right) \\
&  +\frac{1}{3}\left(  x_{3}\left(  p_{3}V_{1}\left(  0,1,2\right)  +\left(
1-p_{3}\right)  V_{1}\left(  2,1,0\right)  \right)  +\left(  1-x_{3}\right)
\left(  p_{2}V_{1}\left(  1,2,0\right)  +\left(  1-p_{2}\right)  V_{1}\left(
1,0,2\right)  \right)  \right)  .
\end{align*}
For example, $\frac{1}{3}x_{1}p_{1}V_{1}\left(  2,0,1\right)  $ is the
probability of:
\[
\text{ $P_{1}$ being selected, \hspace{3mm} $P_{1}$ selecting $P_{2}$
\hspace{3mm} and $P_{1}$ beating $P_{2}$. }%
\]
Substituting the known right hand values, after a considerable amount of
algebra we get
\begin{align*}
V_{1}\left(  1,1,1\right)   &  ={\frac{\left(  1-\mathit{p}_{\mathit{3}%
}\right)  \left(  \mathit{x}_{\mathit{1}}+1-\mathit{x}_{\mathit{2}}\right)
\mathit{p}_{\mathit{1}}}{3\, \left(  {\mathit{p}}_{{\mathit{3}}}%
^{2}-\,{\mathit{p}}_{{\mathit{3}}}+1\right)  }}+{\frac{\mathit{p}_{\mathit{1}%
}\, \left(  \mathit{x}_{\mathit{3}}+1-\mathit{x}_{\mathit{1}}\right)  \left(
1-\mathit{p}_{\mathit{3}}\right)  }{3\left(  \,{\mathit{p}}_{{\mathit{1}}}%
^{2}-\, \mathit{p}_{\mathit{1}}+1\right)  }}\\
&  +{\frac{\left(  1-\mathit{p}_{\mathit{3}}\right)  ^{2}\left(
\mathit{x}_{\mathit{2}}+1-\mathit{x}_{\mathit{3}}\right)  \left(
1-\mathit{p}_{\mathit{2}}\right)  }{3\left(  \,{\mathit{p}}_{{\mathit{3}}}%
^{2}-\, \mathit{p}_{\mathit{3}}+1\right)  }}+{\frac{{\mathit{p}}_{{\mathit{1}%
}}^{2}\left(  \mathit{x}_{\mathit{2}}+1-\mathit{x}_{\mathit{3}}\right)
\mathit{p}_{\mathit{2}}}{3\, \left(  {\mathit{p}}_{{\mathit{1}}}^{2}-\,
\mathit{p}_{\mathit{1}}+1\right)  }.}%
\end{align*}
This can also be written as%

\begin{equation}
V_{1}\left(  1,1,1\right)  =\frac{p_{1}(1-p_{3})(p_{1}+p_{3}-1)(p_{1}%
-p_{3})x_{1}+\Pi_{1}\left(  x_{2},x_{3}\right)  }{3\left(  1-p_{3}+p_{3}%
^{2}\right)  \left(  1-p_{1}+p_{1}^{2}\right)  } \label{eq03001}%
\end{equation}
where $\Pi_{1}\left(  x_{2},x_{3}\right)  $ is a first degree polynomial in
$x_{2}$ and $x_{3}$. Hence we have completed the computation of $\mathbf{V}%
_{1}=\left(  V_{1}\left(  \mathbf{s}\right)  \right)  _{\mathbf{s}\in S}$.

From (\ref{eq03001}) we can easily compute $P_{1}$'s \emph{optimal} strategy.
Since he wants to maximize his payoff by choice of $x_{1}=x_{1}\left(
1,1,1\right)  $, he will use the following strategy:%
\begin{equation}
\widehat{x}_{1}=\left \{
\begin{array}
[c]{cc}%
1\text{ iff:} & \text{either (}p_{1}+p_{3}-1>0\text{ and }p_{1}-p_{3}>0\text{)
or (}p_{1}+p_{3}-1<0\text{ and }p_{1}-p_{3}<0\text{)}\\
0\text{ iff:} & \text{either (}p_{1}+p_{3}-1>0\text{ and }p_{1}-p_{3}<0\text{)
or (}p_{1}+p_{3}-1<0\text{ and }p_{1}-p_{3}>0\text{)}%
\end{array}
\right.  \label{eq03002}%
\end{equation}
For $n\in \left \{  2,3\right \}  $, a similar avalysis yields $P_{n}$'s payoff
$\mathbf{V}_{n}$ (as a function of strategy profile $\mathbf{x=}\left(
x_{1},x_{2},x_{3}\right)  $) and his optimal strategy $\widehat{x}_{n}$. The
expressions for $\widehat{x}_{2}\left(  1,1,1\right)  $ and $\widehat{x}%
_{3}\left(  1,1,1\right)  $ are analogous to (\ref{eq03002}). It is easily
checked that $\left(  \widehat{x}_{1},\widehat{x}_{2},\widehat{x}_{3}\right)
$ is the unique Nash equilibrium of $\Gamma_{3}\left(  \mathbf{p},3\right)  $.

\section{The General Case\label{sec04}}

Let us now study the general case, i.e., when $K\in \left \{  3,4,...\right \}  $.

\subsection{Payoff Equations\label{sec0401}}

For a given stationary strategy profile $\mathbf{x}$, for any $\mathbf{s}%
=\left(  s_{1},s_{2},s_{3}\right)  \in S$ and any $n\in \left \{  1,2,3\right \}
$, let $V_{n}\left(  s_{1},s_{2},s_{3}\right)  =Q_{n}\left(  s_{1},s_{2}%
,s_{3},\mathbf{x}\right)  $ and $\mathbf{V}_{n}=\left(  V_{n}\left(
\mathbf{s}\right)  \right)  _{\mathbf{s}\in S}$. We will now write the
\emph{payoff system}, i.e., the system of equations satisfied by
$\mathbf{V}_{n}=\left(  V_{n}\left(  \mathbf{s}\right)  \right)
_{\mathbf{s}\in S}$.

To this end we start by recognizing that, for a given stationary strategy
profile $\mathbf{x}$, the process $\left(  \mathbf{s}\left(  t\right)
\right)  _{t=0}^{\infty}$ is a Markov chain. Let us define for each
$\mathbf{s}=\left(  s_{1},s_{2},s_{3}\right)  \in S$ the \emph{neighborhood
}$\mathbf{N}\left(  s_{1},s_{2},s_{3}\right)  $ of $\mathbf{s}$ to be the set
of states reachable in one time step from $\mathbf{s}$. For example, when
$\mathbf{s}\in S_{i}$, we have%
\[
\mathbf{N}\left(  s_{1},s_{2},s_{3}\right)  =\left \{  \left(  s_{1}%
+1,s_{2}-1,s_{3}\right)  ,\left(  s_{1}-1,s_{2}+1,s_{3}\right)  ,...,\left(
s_{1},s_{2}-1,s_{3}+1\right)  \right \}  .
\]
For every $\mathbf{s}\in S$ and $\mathbf{s}^{\prime}\in \mathbf{N}\left(
\mathbf{s}\right)  $ we define the transition probabilities%
\[
\pi_{\mathbf{s}^{\prime},\mathbf{s}}=\Pr \left(  \mathbf{s}\left(  t+1\right)
=\mathbf{s|s}\left(  t\right)  =\mathbf{s}^{\prime}\right)  .
\]
Then $\mathbf{V}_{1}$ satisfies the following equations.

\begin{enumerate}
\item For boundary states we have:%
\begin{equation}
V_{1}\left(  K,0,0\right)  =1; \label{eq04001}%
\end{equation}
for all $s_{2}\in \left(  0,1,...,K\right)  $:
\begin{equation}
V_{1}\left(  0,s_{2},K-s_{2}\right)  =0; \label{eq04002}%
\end{equation}
and for all $s_{1}\in \left(  1,...,K-1\right)  $:%
\begin{align*}
V_{1}\left(  s_{1},K-s_{1},0\right)   &  =\sum_{\mathbf{s}^{\prime}%
\in \mathbf{N}\left(  s_{1},K-s_{1},0\right)  }\pi_{\left(  s_{1}%
,K-s_{1},0\right)  ,\mathbf{s}^{\prime}}V_{1}\left(  \mathbf{s}^{\prime
}\right)  ,\\
V_{1}\left(  s_{1},0,K-s_{1}\right)   &  =\sum_{\mathbf{s}^{\prime}%
\in \mathbf{N}\left(  s_{1},0,K-s_{1}\right)  }\pi_{\left(  s_{1}%
,0,K-s_{1}\right)  ,\mathbf{s}^{\prime}}V_{1}\left(  \mathbf{s}^{\prime
}\right)  ,
\end{align*}
which becomes
\begin{align}
V_{1}\left(  s_{1},K-s_{1},0\right)   &  =p_{1}V_{1}\left(  s_{1}%
+1,K-s_{1}-1,0\right)  +\left(  1-p_{1}\right)  V_{1}\left(  s_{1}%
-1,K-s_{1}+1,0\right) \label{eq04003}\\
V_{1}\left(  s_{1},0,K-s_{1}\right)   &  =\left(  1-p_{3}\right)  V_{1}\left(
s_{1}+1,0,K-s_{1}-1\right)  +p_{1}V_{1}\left(  s_{1}-1,0,K-s_{1}+1\right)  .
\label{eq04004}%
\end{align}

\item For all interior states $\left(  s_{1},s_{2},s_{3}\right)  $ $\in
S_{i}\ $we have
\[
V_{1}\left(  s_{1},s_{2},s_{3}\right)  =\sum_{\mathbf{s}^{\prime}\in
\mathbf{N}\left(  s_{1},s_{2},s_{3}\right)  }\pi_{\left(  s_{1},s_{2}%
,s_{3}\right)  ,\mathbf{s}^{\prime}}V_{1}\left(  \mathbf{s}^{\prime}\right)
\]

\end{enumerate}

which can be written as either
\begin{align}
V_{1}\left(  s_{1},s_{2},s_{3}\right)   &  =\left(  \frac{x_{1}\left(
s_{1},s_{2},s_{3}\right)  +\left(  1-x_{2}\left(  s_{1},s_{2},s_{3}\right)
\right)  }{3}\right)  p_{1}V_{1}\left(  s_{1}+1,s_{2}-1,s_{3}\right)
\nonumber \\
&  +\left(  \frac{x_{1}\left(  s_{1},s_{2},s_{3}\right)  +\left(
1-x_{2}\left(  s_{1},s_{2},s_{3}\right)  \right)  }{3}\right)  \left(
1-p_{1}\right)  V_{1}\left(  s_{1}-1,s_{2}+1,s_{3}\right) \nonumber \\
&  +\left(  \frac{x_{2}\left(  s_{1},s_{2},s_{3}\right)  +\left(
1-x_{3}\left(  s_{1},s_{2},s_{3}\right)  \right)  }{3}\right)  p_{2}%
V_{1}\left(  s_{1},s_{2}+1,s_{3}-1\right) \nonumber \\
&  +\left(  \frac{x_{2}\left(  s_{1},s_{2},s_{3}\right)  +\left(
1-x_{3}\left(  s_{1},s_{2},s_{3}\right)  \right)  }{3}\right)  \left(
1-p_{2}\right)  V_{1}\left(  s_{1},s_{2}-1,s_{3}+1\right) \nonumber \\
&  +\left(  \frac{x_{3}\left(  s_{1},s_{2},s_{3}\right)  +\left(
1-x_{1}\left(  s_{1},s_{2},s_{3}\right)  \right)  }{3}\right)  p_{3}%
V_{1}\left(  s_{1}-1,s_{2},s_{3}+1\right) \nonumber \\
&  +\left(  \frac{x_{3}\left(  s_{1},s_{2},s_{3}\right)  +\left(
1-x_{1}\left(  s_{1},s_{2},s_{3}\right)  \right)  }{3}\right)  \left(
1-p_{3}\right)  V_{1}\left(  s_{1}+1,s_{2},s_{3}-1\right)  , \label{eq04006}%
\end{align}

\noindent or as%

\begin{align}
V_{1}\left(  s_{1},s_{2},s_{3}\right)   &  =\frac{x_{1}\left(  s_{1}%
,s_{2},s_{3}\right)  }{3}\left(  p_{1}V_{1}\left(  s_{1}+1,s_{2}%
-1,s_{3}\right)  +\left(  1-p_{1}\right)  V_{1}\left(  s_{1}-1,s_{2}%
+1,s_{3}\right)  \right) \nonumber \\
&  +\frac{\left(  1-x_{1}\left(  s_{1},s_{2},s_{3}\right)  \right)  }%
{3}\left(  \left(  1-p_{3}\right)  V_{1}\left(  s_{1}+1,s_{2},s_{3}-1\right)
+p_{3}V_{1}\left(  s_{1}-1,s_{2},s_{3}+1\right)  \right) \nonumber \\
&  +\frac{x_{2}\left(  s_{1},s_{2},s_{3}\right)  }{3}\left(  p_{2}V_{1}\left(
s_{1},s_{2}+1,s_{3}-1\right)  +\left(  1-p_{2}\right)  V_{1}\left(
s_{1},s_{2}-1,s_{3}+1\right)  \right) \nonumber \\
&  +\frac{\left(  1-x_{2}\left(  s_{1},s_{2},s_{3}\right)  \right)  }%
{3}\left(  \left(  1-p_{1}\right)  V_{1}\left(  s_{1}-1,s_{2}+1,s_{3}\right)
+p_{1}V_{1}\left(  s_{1}+1,s_{2}-1,s_{3}\right)  \right) \nonumber \\
&  +\frac{x_{3}\left(  s_{1},s_{2},s_{3}\right)  }{3}\left(  p_{3}V_{1}\left(
s_{1}-1,s_{2},s_{3}+1\right)  +\left(  1-p_{3}\right)  V_{1}\left(
s_{1}+1,s_{2},s_{3}-1\right)  \right) \nonumber \\
&  +\frac{\left(  1-x_{3}\left(  s_{1},s_{2},s_{3}\right)  \right)  }%
{3}\left(  \left(  1-p_{2}\right)  V_{1}\left(  s_{1},s_{2}-1,s_{3}+1\right)
+p_{2}V_{1}\left(  s_{1},s_{2}+1,s_{3}-1\right)  \right)  . \label{eq04005}%
\end{align}
For every $\left(  s_{1},s_{2},s_{3}\right)  \in S_{i}$, the values of the
$\pi_{\left(  s_{1},s_{2},s_{3}\right)  ,\mathbf{s}^{\prime}}$ probabilities
can be read from (\ref{eq04006}); for $\left(  s_{1},s_{2},s_{3}\right)  \in
S_{1}$, the values of the $\pi_{\left(  s_{1},s_{2},s_{3}\right)
,\mathbf{s}^{\prime}}$ probabilities can \ be read from (\ref{eq04003}%
\noindent)-(\ref{eq04004}).

The equations (\ref{eq04001}\noindent)-(\ref{eq04005}) are the \emph{payoff
system} for $\mathbf{V}_{1}$. We have similar payoff systems for
$\mathbf{V}_{2}$ and $\mathbf{V}_{3}$. In the rest of this section we will
focus on $\mathbf{V}_{1}$. Let $\Pi \left(  \mathbf{x}\right)  $ be the
transition probability matrix with $\Pi_{\mathbf{s,s}^{\prime}}\left(
\mathbf{x}\right)  =\pi_{\mathbf{s,s}^{\prime}}\left(  \mathbf{x}\right)  $
(for all $\mathbf{s,s}^{\prime}\in S$). Then the \emph{payoff system} can be
written as
\begin{equation}
V_{1}\left(  K,0,0\right)  =1\text{ and }V_{1}\left(  0,K,0\right)
=V_{1}\left(  K,0,0\right)  =0\text{ and }\mathbf{V}_{1}=\Pi \left(
\mathbf{x}\right)  \mathbf{V}_{1}\text{. } \label{eq04007}%
\end{equation}
This is a linear system and can be solved by the standard methods. An
alternative but equivalent form of (\ref{eq04007}) is the following. Recall
that we have numbered the states $\mathbf{s}\in S$ so that the first state is
$\left(  K,0,0\right)  $, the second is $\left(  0,K,0\right)  $ and the third
is $\left(  0,0,K\right)  $. Now define the matrix $\widetilde{\Pi}_{1}\left(
\mathbf{x}\right)  $ and the vector $\mathbf{b}_{1}$ as follows:%
\[
\widetilde{\Pi}_{1}\left(  \mathbf{x}\right)  =\left[
\begin{array}
[c]{cccc}%
0 & 0 & ... & 0\\
0 & 0 & ... & 0\\
0 & 0 & ... & 0\\
\pi_{4,1} & \pi_{4,2} & ... & \pi_{4,N_{K}}\\
... & ... & ... & ...\\
\pi_{N_{K},1} & \pi_{N_{K},2} & ... & \pi_{N_{K},N_{K}}%
\end{array}
\right]  ,\qquad \mathbf{b}_{1}=\left[
\begin{array}
[c]{c}%
1\\
0\\
0\\
0\\
...\\
0
\end{array}
\right]
\]
Then (\ref{eq04007}) is equivalent to
\begin{equation}
\mathbf{V}_{1}=\widetilde{\Pi}_{1}\left(  \mathbf{x}\right)  \mathbf{V}%
_{1}+\mathbf{b}_{1}\mathbf{.} \label{eq04008b}%
\end{equation}
If the payoff system (\ref{eq04008b}) has exactly one solution $\overline
{\mathbf{V}}_{1}$, then $\overline{V}_{1}\left(  s\right)  $ is $P_{1}$'s
payoff when the game starts with capitals $s=\left(  s_{1},s_{2},s_{3}\right)
$.

\subsection{Markov Chains\label{sec0402}}

As already mentioned, for every stationary strategy profile $\mathbf{x}$, the
process $\left(  \mathbf{s}\left(  t\right)  \right)  _{t=0}^{\infty}$ is a
Markov chain, with transition probability matrix $\Pi \left(  \mathbf{x}%
\right)  $.

\begin{proposition}
\normalfont \label{prop04002}If for all $m,n\in \left \{  1,2,3\right \}  $ we
have $p_{mn}\in \left(  0,1\right)  $ then, for every $\mathbf{x}$, every
interior state communicates with all terminal states.
\end{proposition}

\begin{proof}
Take any state $\mathbf{s=}\left(  s_{1},s_{2},s_{3}\right)  \in S_{i}$ and
suppose that the game is in $\mathbf{s}$; let $\delta=\min_{m,n}p_{mn}$. Then
we have%
\begin{align*}
\Pr \left(  \text{\textquotedblleft}P_{1}\text{ is selected and he loses the
turn\textquotedblright}\right)   &  \geq \frac{1}{3}\left(  x_{12}\left(
\mathbf{s}\right)  p_{21}+x_{13}\left(  \mathbf{s}\right)  p_{31}\right) \\
&  \geq \frac{1}{3}\left(  x_{1}\left(  \mathbf{s}\right)  p_{21}+\left(
1-x_{1}\left(  \mathbf{s}\right)  \right)  p_{31}\right)  \geq \frac{\delta}%
{3}.
\end{align*}
In other words, there is a positive probability that the game moves to a state
$\mathbf{s}^{\prime}=\left(  s_{1}-1,s_{2}^{\prime},s_{3}^{\prime}\right)  $.
Hence there is a positive probability that eventually the game will reach a
boundary state $\mathbf{s}^{\prime \prime}=\left(  0,s_{2}^{\prime}%
,s_{3}^{\prime}\right)  $; either $\mathbf{s}^{\prime \prime}\ $is a terminal
state itself, or $\mathbf{s}^{\prime \prime}\in S\backslash \left(  S_{i}\cup
S_{\tau}\right)  $ and communicates with both $\left(  0,K,0\right)  $ and
$\left(  0,0,K\right)  $. Hence $\mathbf{s}$ communicates with both $\left(
0,K,0\right)  $ and $\left(  0,0,K\right)  $. Repeating the argument for
$P_{2}$ we see that $\mathbf{s}$ also communicates with $\left(  K,0,0\right)
$.
\end{proof}

\noindent By Proposition \ref{prop04002} and standard Markov chain results
\cite{GrinsteadSnellProb,KemenySnellFMC}, every nonterminal state is transient
and the game will reach \emph{some} terminal state w.p. 1 and in finite
expected time. Furthermore we have the following.

\begin{proposition}
\normalfont \label{prop04003}If for all $m,n\in \left \{  1,2,3\right \}  $ we
have $p_{mn}\in \left(  0,1\right)  $, then, for every $\mathbf{x}$, the limit
$\overline{\Pi}\left(  \mathbf{x}\right)  =\lim_{t\rightarrow \infty}\Pi
^{t}\left(  \mathbf{x}\right)  $ exists.
\end{proposition}

\begin{proof}
Recall that the three terminal states are numbered as first, second and third.
Since they are absorbing and all other states are transient, the probability
transition matrix has the form%
\begin{equation}
\Pi \left(  \mathbf{x}\right)  =\left[
\begin{array}
[c]{cc}%
I & 0\\
W & U
\end{array}
\right]  . \label{eq04009}%
\end{equation}
Then we have
\[
\Pi^{t}\left(  \mathbf{x}\right)  =\left[
\begin{array}
[c]{cc}%
I & 0\\
\left(  \sum_{i=0}^{t}U^{i}\right)  W & U^{t}%
\end{array}
\right]
\]
By standard Markov chain results \cite{GrinsteadSnellProb,KemenySnellFMC}
$\  \lim_{t\rightarrow \infty}U^{t}=0$ and $\lim_{t\rightarrow \infty}\sum
_{i=0}^{t}U^{i}$ exists and equals $\left(  I-U\right)  ^{-1}$. Hence
\[
\lim_{t\rightarrow \infty}\Pi^{t}\left(  \mathbf{x}\right)  =\left[
\begin{array}
[c]{cc}%
I & 0\\
\lim_{t\rightarrow \infty}\left(  \sum_{i=0}^{t}U^{i}\right)  W & U^{t}%
\end{array}
\right]  =\left[
\begin{array}
[c]{cc}%
I & 0\\
\left(  I-U\right)  ^{-1}W & 0
\end{array}
\right]  .
\]
This completes the proof.
\end{proof}

\begin{proposition}
\normalfont \label{prop04004}If the limit $\overline{\Pi}\left(
\mathbf{x}\right)  =\lim_{t\rightarrow \infty}\Pi^{t}\left(  \mathbf{x}\right)
$ exists then%
\[
\forall \left(  s_{1},s_{2},s_{3}\right)  \in S:\overline{\Pi}_{\left(
s_{1},s_{2},s_{3}\right)  ,\left(  K,0,0\right)  }\left(  \mathbf{x}\right)
=V_{1}\left(  s_{1},s_{2},s_{3}\right)  .
\]

\end{proposition}

\begin{proof}
\noindent Recall that
\[
\forall \left(  s_{1},s_{2},s_{3}\right)  \in S:V_{1}\left(  s_{1},s_{2}%
,s_{3}\right)  =\Pr \left(  \text{\textquotedblleft}P_{1}\text{ wins when the
game has started in }\left(  s_{1},s_{2},s_{3}\right)
\text{\textquotedblright}\right)
\]
We have the following possibilities.

\begin{enumerate}
\item If $\left(  s_{1},s_{2},s_{3}\right)  =\left(  K,0,0\right)  $, then
$\overline{\Pi}_{\left(  s_{1},s_{2},s_{3}\right)  ,\left(  K,0,0\right)
}\left(  \mathbf{x}\right)  =1=V_{1}\left(  K,0,0\right)  $.

\item If $\left(  s_{1},s_{2},s_{3}\right)  =\left(  0,K,0\right)  $, then
$\overline{\Pi}_{\left(  0,K,0\right)  ,\left(  K,0,0\right)  }\left(
\mathbf{x}\right)  =0=V_{1}\left(  0,K,0\right)  $.

\item If $\left(  s_{1},s_{2},s_{3}\right)  =\left(  0,0,K\right)  $, then
$\overline{\Pi}_{\left(  0,0,K\right)  ,\left(  K,0,0\right)  }\left(
\mathbf{x}\right)  =0=V_{1}\left(  0,0,K\right)  $.

\item If $\left(  s_{1},s_{2},s_{3}\right)  \in S\backslash S_{\tau}$, then
$\overline{\Pi}_{\left(  s_{1},s_{2},s_{3}\right)  ,\left(  K,0,0\right)
}\left(  \mathbf{x}\right)  =\left(  \left(  \sum_{i=0}^{\infty}U^{i}\right)
W\right)  _{\left(  s_{1},s_{2},s_{3}\right)  ,\left(  K,0,0\right)  }$. This
is the probability that: $\left(  s_{1}\left(  t\right)  ,s_{2}\left(
t\right)  ,s_{3}\left(  t\right)  \right)  $ stays in $S\backslash S_{\tau}$
for $i$ steps (for some $i\in \left(  0,1,...\right)  $) and then moves into
$\left(  K,0,0\right)  $. In every such case $P_{1}$ wins.
\end{enumerate}

\noindent This completes the proof.
\end{proof}

\bigskip

\noindent Hence we can compute $\mathbf{V}_{1}$ as the limit of $\Pi
^{t}\left(  \mathbf{x}\right)  $. From standard MC\ analysis
\cite{GrinsteadSnellProb,KemenySnellFMC} we can get estimates of the $\Pi^{t}$
rate of convergence, mean absorption time etc. Another way to compute
$\mathbf{V}_{1}$ is provided by the following.

\begin{proposition}
\normalfont Suppose that $\mathbf{x}$ is such that $\overline{\Pi}\left(
\mathbf{x}\right)  =\lim_{t\rightarrow \infty}\Pi^{t}\left(  \mathbf{x}\right)
$ exists. Choose$\ $an arbitrary $N_{K}\times1$ vector $\mathbf{U}%
_{1}^{\left(  0\right)  }$ and
\begin{equation}
\text{For}\ t\in \left \{  0,1,2...\right \}  \text{ let:}\  \mathbf{U}%
_{1}^{\left(  t+1\right)  }=\widetilde{\Pi}_{1}\left(  \mathbf{x}\right)
\mathbf{U}_{1}^{\left(  t\right)  }+\mathbf{b}_{1}\mathbf{.} \label{eq04008a}%
\end{equation}
The iteration (\ref{eq04008a}) converges and $\overline{\mathbf{U}}%
=\lim_{t\rightarrow \infty}\mathbf{U}^{\left(  t\right)  }$ is a solution of
the payoff system (\ref{eq04008b}).
\end{proposition}

\begin{proof}
\noindent The key fact here is that, for any solution $\mathbf{V}_{1}$ of the
payoff system:

\begin{enumerate}
\item when $\mathbf{s\in}S_{\tau}$, $V_{1}\left(  s\right)  $ is fixed to be
either 0 or 1;

\item when $\mathbf{s\in}S_{\tau}^{c}$, $V_{1}\left(  \mathbf{s}\right)  $ is
the weighted average of its neighbors:
\[
V_{1}\left(  \mathbf{s}\right)  =\sum_{\mathbf{s}^{\prime}\in N\left(
\mathbf{s}\right)  }\pi_{\mathbf{s,s}^{\prime}}V_{1}\left(  \mathbf{s}%
^{\prime}\right)  .
\]

\end{enumerate}

\noindent Consequently, every solution of the payoff system is a harmonic
function. Hence the payoff system has a unique solution $\mathbf{V}_{1}$ which
can be obtained as the limit of (\ref{eq04008a}). \noindent
\end{proof}

\subsection{Stationary Equilibria\label{sec0403}}

We want to show that $\Gamma_{3}\left(  \mathbf{p},K\right)  $ posseses a Nash
equilibrium (NE), i.e., a strategy profile $\widehat{\mathbf{\sigma}}$ such
that:%
\begin{equation}
\forall n,\forall \mathbf{\sigma}_{n}:Q_{n}\left(  \mathbf{s}\left(  0\right)
,\widehat{\mathbf{\sigma}}_{n},\widehat{\mathbf{\sigma}}_{-n}\right)  \geq
Q_{n}\left(  \mathbf{s}\left(  0\right)  ,\mathbf{\sigma}_{n},\widehat
{\mathbf{\sigma}}_{-n}\right)  \label{eq04031}%
\end{equation}
To this end we introduce an auxiliary \emph{discounted }game, denoted by
$\widetilde{\Gamma}_{3}\left(  \mathbf{p},K,\gamma \right)  $, where $\gamma
\in \left(  0,1\right)  $. $\widetilde{\Gamma}_{3}\left(  \mathbf{p}%
,K,\gamma \right)  $ is the same as $\Gamma_{3}\left(  \mathbf{p},K\right)  $
except for the following.

\begin{enumerate}
\item In $\widetilde{\Gamma}_{3}\left(  \mathbf{p},K,\gamma \right)  $ the
states $\left(  K,0,0\right)  $, $\left(  0,K,0\right)  $, $\left(
0,0,K\right)  $ are \emph{preterminal} and we introduce a new terminal state
$\overline{s}$; state transition probabilities remain the same except that
every preterminal state $s\in S_{\tau}$ transits to $\overline{s}$ with
probability one.

\item We define \emph{turn payoff} functions $q_{n}\left(  s\right)  $:%
\[
\forall n:q_{n}\left(  s\right)  =\left \{
\begin{array}
[c]{rl}%
1 & \text{iff }s\in S_{n},\\
0 & \text{else.}%
\end{array}
\right.
\]

\item The total payoff function $\widetilde{\mathsf{Q}}_{n}$ and the expected
total payoff function $\widetilde{Q}_{n}$ are defined as follows:%
\[
\widetilde{\mathsf{Q}}_{n}\left(  \mathbf{h}\right)  =\sum_{t=0}^{\infty
}\gamma^{t}q_{n}\left(  s\right)  ,\qquad \widetilde{Q}_{n}\left(  s\left(
0\right)  ,\sigma \right)  =\mathbb{E}_{\mathbf{s}\left(  0\right)
,\mathbf{\sigma}}\left(  \widetilde{\mathsf{Q}}_{n}\left(  \mathbf{h}\right)
\right)
\]

\end{enumerate}

$\bigskip$

\noindent$\widetilde{\Gamma}_{3}\left(  \mathbf{p},K,\gamma \right)  $ is a
\emph{discounted} stochastic game and, according to the following well known
theorem, possesses a stationary NE.

\begin{proposition}
[Fink \cite{Fink1961}]\normalfont Every $N$-player discounted stochastic game
has a stationary Nash equilibrium.
\end{proposition}

\noindent In the following proposition we strengthen Fink's theorem for the
case of $\widetilde{\Gamma}_{3}\left(  \mathbf{p},K,\gamma \right)  $.

\begin{proposition}
\label{prop0401}\normalfont For every $\mathbf{p}$, $K$ and $\gamma$,
$\widetilde{\Gamma}_{3}\left(  \mathbf{p},K,\gamma \right)  $ has a
deterministic stationary Nash equilibrium. In other words, there exists a
$\widehat{\mathbf{x}}=\left(  \widehat{\mathbf{x}}_{1},\widehat{\mathbf{x}%
}_{2},\widehat{\mathbf{x}}_{3}\right)  $ such that%
\begin{equation}
\forall n\in \left \{  1,2,3\right \}  ,\forall \mathbf{s}_{0}\in S,\forall
\mathbf{x}^{n}:\widetilde{Q}^{n}\left(  \mathbf{s}_{0},\widehat{\mathbf{x}%
}_{n},\widehat{\mathbf{x}}_{-n}\right)  \geq \widetilde{Q}^{n}\left(
\mathbf{s}_{0},\mathbf{x}_{n},\widehat{\mathbf{x}}_{-n}\right)  .
\label{eq02011}%
\end{equation}

\end{proposition}

\begin{proof}
According to \cite{Fink1961}, for a general $N$-player discounted stochastic
game, the following equations must be satisfied for all $n$ and $\mathbf{s}$
at a Nash equilibrium\footnote{We have modified Fink's notation, so as to
conform to our own.}:
\begin{equation}
\mathcal{V}_{n}\left(  \mathbf{s}\right)  =\max_{\mathbf{\sigma}_{n}\left(
\mathbf{s}\right)  }\sum_{a_{1}}...\sum_{a_{N}}\sigma_{1,a_{1}}\left(
\mathbf{s}\right)  ...\sigma_{N,a_{N}}\left(  \mathbf{s}\right)  \left[
q_{n}\left(  \mathbf{s}\right)  +\gamma \sum_{\mathbf{s}^{\prime}}%
\mathcal{P}\left(  \mathbf{s}^{\prime}|\mathbf{s},a_{1},...,a_{N}\right)
\mathcal{V}_{n}\left(  \mathbf{s}^{\prime}\right)  \right]  , \label{04051}%
\end{equation}
where:

\begin{enumerate}
\item $\sigma_{n}$ is $P_{n}$'s strategy, $\sigma_{n}\left(  \mathbf{s}%
\right)  $ is his action probability vector at state $\mathbf{s}$ and
$\sigma_{n,a_{n}}\left(  \mathbf{s}\right)  $ is the probability that, given
the current game state is $\mathbf{s}$, $P_{n}$ plays action $a_{n}$.

\item $\mathcal{V}_{n}\left(  \mathbf{s}\right)  $ is the expected payoff to
$P_{n}$ obtained at equilibrium, when the game has started at state
$\mathbf{s}$ and the strategy and the strategy profile $\sigma=\left(
\sigma_{1},...,\sigma_{N}\right)  $ is used.

\item $q_{n}\left(  \mathbf{s}\right)  $ is $P_{n}$'s turn payoff at state
$\mathbf{s}$.

\item $\mathcal{P}\left(  \mathbf{s}^{\prime}|\mathbf{s},a_{1},...,a_{N}%
\right)  $ is the probability that, given that the current state is
$\mathbf{s}$ and the player actions are $a_{1},...,a_{N}$, the next state is
$\mathbf{s}^{\prime}$.
\end{enumerate}

Let us now consider $\widetilde{\Gamma}_{3}\left(  \mathbf{p},K,\gamma \right)
$. Since no player has a choice of strategy on boundary states, we only need
to consider interior states $\mathbf{s}=\left(  s_{1},s_{2},s_{3}\right)  \in
S_{i}$. Letting $n=1$ and taking into account (\ref{eq04005}), (\ref{04051})
becomes%
\[
V_{1}\left(  s_{1},s_{2},s_{3}\right)  =\max_{x_{1}\left(  s_{1},s_{2}%
,s_{3}\right)  }\gamma G\left(  x_{1}\left(  s_{1},s_{2},s_{3}\right)
\right)
\]
where%
\begin{align*}
G\left(  x_{1}\left(  s_{1},s_{2},s_{3}\right)  \right)   &  =x_{1}\left(
s_{1},s_{2},s_{3}\right)  \frac{p_{1}V_{1}\left(  s_{1}+1,s_{2}-1,s_{3}%
\right)  +\left(  1-p_{1}\right)  V_{1}\left(  s_{1}-1,s_{2}+1,s_{3}\right)
}{3}\\
&  +\left(  1-x_{1}\left(  s_{1},s_{2},s_{3}\right)  \right)  \frac{\left(
1-p_{3}\right)  V_{1}\left(  s_{1}+1,s_{2},s_{3}-1\right)  +p_{3}V_{1}\left(
s_{1}-1,s_{2},s_{3}+1\right)  }{3}\\
&  +\text{terms which do not involve }x_{1}\left(  s_{1}s_{2}s_{3}\right)
\end{align*}
Clearly, the maximum is achieved at either $x_{1}\left(  s_{1},s_{2}%
,s_{3}\right)  =1$ or $x_{1}\left(  s_{1},s_{2},s_{3}\right)  =0$. This holds
for every interior state, hence $P_{1}$'s strategy at equilibrium is
deterministic (or can be substituted by an equivalent deterministic one).

\vspace{3mm}

\noindent The same analysis can be applied to $P_{2}$ and $P_{3}$ to complete
the proof.
\end{proof}

\bigskip

\hspace{10mm}

\begin{corollary}
\label{prop0404}\normalfont For every $\mathbf{p}$ and $K$, $\Gamma_{3}\left(
\mathbf{p},K\right)  $ has a deterministic stationary NE.
\end{corollary}

\begin{proof}
Let $\widehat{\mathbf{x}}\mathbf{=}\left(  \widehat{\mathbf{x}}_{1}%
,\widehat{\mathbf{x}}_{2},\widehat{\mathbf{x}}_{3}\right)  $ be a a
deterministic stationary NE of $\widetilde{\Gamma}_{3}\left(  \mathbf{p}%
,K,\gamma \right)  $. Suppose that, for some $\mathbf{s}_{0}$ and $t$, we have
\[
\widetilde{Q}\left(  \mathbf{s}_{0},\widehat{\mathbf{x}}\right)  =\left(
\widetilde{Q}_{1}\left(  \mathbf{s}_{0},\widehat{\mathbf{x}}\right)
,\widetilde{Q}_{2}\left(  \mathbf{s}_{0},\widehat{\mathbf{x}}\right)
,\widetilde{Q}_{3}\left(  \mathbf{s}_{0},\widehat{\mathbf{x}}\right)  \right)
=\left(  \gamma^{t},0,0\right)  .
\]
This means that, in $\widetilde{\Gamma}_{3}\left(  \mathbf{p},K,\gamma \right)
$, $P_{1}$ has won at time $t$. Applying the same $\widehat{\mathbf{x}}$ in
$\Gamma_{3}\left(  \mathbf{p},K,\gamma \right)  $, we have%
\[
Q\left(  \mathbf{s}_{0},\widehat{\mathbf{x}}\right)  =\left(  Q_{1}\left(
\mathbf{s}_{0},\widehat{\mathbf{x}}\right)  ,Q_{2}\left(  \mathbf{s}%
_{0},\widehat{\mathbf{x}}\right)  ,Q_{3}\left(  \mathbf{s}_{0},\widehat
{\mathbf{x}}\right)  \right)  =\left(  1,0,0\right)  .
\]
Clearly%
\[
\forall \mathbf{x}_{1}:Q_{1}\left(  \mathbf{s}_{0},\widehat{\mathbf{x}}%
_{1},\widehat{\mathbf{x}}_{-1}\right)  \geq Q_{1}\left(  \mathbf{s}%
_{0},\mathbf{x}_{1},\widehat{\mathbf{x}}_{-1}\right)  .
\]
Also, if there was some $\mathbf{y}_{n}$ (with $n\in \left \{  2,3\right \}  $)
such that
\[
Q_{1}\left(  \mathbf{s}_{0},\widehat{\mathbf{x}}_{n},\widehat{\mathbf{x}}%
_{-n}\right)  <Q_{1}\left(  \mathbf{s}_{0},\mathbf{y}_{n},\widehat{\mathbf{x}%
}_{-n}\right)
\]
then $P_{n}$ could win $\Gamma_{3}\left(  \mathbf{p},K\right)  $ (when
starting at $\mathbf{s}_{0}$) by using $\mathbf{y}_{n}$ against $\widehat
{\mathbf{x}}_{-n}$. But then he could also win $\widetilde{\Gamma}_{3}\left(
\mathbf{p},K,\gamma \right)  $ (when starting at $\mathbf{s}_{0}$) by using the
same $\mathbf{y}_{n}$ against $\widehat{\mathbf{x}}_{-n}$, so we would have%
\[
\widetilde{Q}_{1}\left(  \mathbf{s}_{0},\widehat{\mathbf{x}}_{n}%
,\widehat{\mathbf{x}}_{-n}\right)  <\widetilde{Q}_{1}\left(  \mathbf{s}%
_{0},\mathbf{y}_{n},\widehat{\mathbf{x}}_{-n}\right)  .
\]
But this is contrary to the hypothesis. The proof is similar for the case in
which, for some $\mathbf{s}_{0}$ and $t$,
\[
\widetilde{Q}\left(  \mathbf{s}_{0},\widehat{\mathbf{x}}\right)  =\left(
\widetilde{Q}_{1}\left(  \mathbf{s}_{0},\widehat{\mathbf{x}}\right)
,\widetilde{Q}_{2}\left(  \mathbf{s}_{0},\widehat{\mathbf{x}}\right)
,\widetilde{Q}_{3}\left(  \mathbf{s}_{0},\widehat{\mathbf{x}}\right)  \right)
=\left(  0,0,0\right)  .
\]
We conclude that:%
\[
\forall n\in \left \{  1,2,3\right \}  ,\forall \mathbf{s}_{0}\in S,\forall
\mathbf{x}^{n}:\widetilde{Q}_{n}\left(  s_{0},\widehat{\mathbf{x}}%
_{n},\widehat{\mathbf{x}}_{-n}\right)  \geq \widetilde{Q}_{n}\left(
s_{0},\mathbf{x}_{n},\widehat{\mathbf{x}}_{-n}\right)
\]
and the proof is complete.
\end{proof}

\subsection{Nonstationary Equilibria\label{sec0404}}

So far we have established that $\Gamma_{3}\left(  \mathbf{p},K\right)  $ has
at least one stationary NE. It is conceivable that the existence of
\emph{nonstationary} NE can also be established by standard game theoretical
arguments. We will only present some initial ideas on the subject, and
relegate a full analysis to the future.

We limit our discussion to a variant of the discounted game, which we call
$\Gamma_{3}^{\prime}\left(  \mathbf{p},K,\gamma \right)  $. This is the same as
$\widetilde{\Gamma}_{3}\left(  \mathbf{p},K,\gamma \right)  $, except in that
the losing players incur a cost of $-1$; hence a player who loses all his
capital at some time $t^{\prime}$, receives payoff $-\gamma^{t^{\prime}}$.

In $\Gamma_{3}^{\prime}\left(  \mathbf{p},K,\gamma \right)  $ a player wants to
win or, if this is not possible, to lose at the latest possible time. Under
the circumstances, it is conceivable that, for example, $P_{1}$ and $P_{2}$
form an alliance against $P_{3}$, with the goal of bankrupting him first. For
example, $P_{1}$ and $P_{2}$ might continuously play against $P_{3}$, forming
in effect a \  \textquotedblleft superplayer\textquotedblright \ with a higher
probability of defeating $P_{3}$. However, what is required is a mechanism to
enforce the alliance. This can be provided by using a \textquotedblleft
trigger\textquotedblright \ strategy: $P_{1}$ selects $P_{3}$ as long as
$P_{2}$ does the same; if at some round $P_{2}$ selects $P_{1}$, then $P_{1}$
selects $P_{2}$ in all subsequent rounds. It appears likely that, under
appropriate conditions on $\gamma$, this strategy will be a
(nonstationary)\ NE. A \textquotedblleft tit-for-tat\textquotedblright%
\ strategy might work in similar manner.

\section{Computation of Stationary Equilibria\label{sec05}}

We now present a very preliminary examination of stationary NE computation (a
more extensive account will be presented at a later time). We want to find a
stationary strategy profile $\widehat{\mathbf{x}}$ such that:%
\begin{equation}
\forall n,\mathbf{s},\mathbf{x}_{n}:Q_{n}\left(  \mathbf{s},\widehat
{\mathbf{x}}_{n},\widehat{\mathbf{x}}_{-n}\right)  \geq Q_{n}\left(
\mathbf{s},\mathbf{x}_{n},\widehat{\mathbf{x}}_{-n}\right)  \label{eq04041}%
\end{equation}
or, equivalently,
\begin{equation}
\forall n,\mathbf{s},\mathbf{x}_{n}:V_{n,\mathbf{s}}\left(  \widehat
{\mathbf{x}}_{n},\widehat{\mathbf{x}}_{-n}\right)  \geq V_{n,\mathbf{s}%
}\left(  \mathbf{x}_{n},\widehat{\mathbf{x}}_{-n}\right)  . \label{eq04042}%
\end{equation}
Now, (\ref{eq04042}) can be rewritten as
\begin{equation}
\forall n,\mathbf{s}:V_{n,s}\left(  \widehat{\mathbf{x}}_{n},\widehat
{\mathbf{x}}_{-n}\right)  =\max_{\mathbf{x}_{n}}V_{n,s}\left(  \mathbf{x}%
_{n},\widehat{\mathbf{x}}_{-n}\right)  =\max_{\mathbf{x}_{n}}\left(  \left[
\widetilde{\Pi}\left(  \mathbf{x}_{n},\widehat{\mathbf{x}}_{-n}\right)
\mathbf{V}_{n}+\mathbf{u}_{n}\right]  _{s}\right)  . \label{eq04033}%
\end{equation}
We define
\begin{equation}
\forall n,\mathbf{s}:F_{n,\mathbf{s}}\left(  \mathbf{V}_{1}|\mathbf{x}\right)
=\max_{\mathbf{x}_{n}}\left(  \left[  \widetilde{\Pi}\left(  \mathbf{x}%
\right)  \mathbf{V}_{n}+\mathbf{u}_{n}\right]  _{s}\right)  \label{eq04044}%
\end{equation}
Then we can rewrite (\ref{eq04033}) as%
\begin{equation}
\forall n:\mathbf{V}_{n}\left(  \widehat{\mathbf{x}}_{n},\widehat{\mathbf{x}%
}_{-n}\right)  =\mathbf{F}_{n}\left(  \mathbf{V}_{n}|\widehat{\mathbf{x}%
}\right)  , \label{eq04045}%
\end{equation}
This system of nonlinear equations will, by the previous analysis, have at
least one solution. We can either solve (\ref{eq04045}) directly or,
alternatively, minimize%
\begin{equation}
J\left(  \mathbf{x}\right)  =\sum_{n=1}^{3}\left \Vert \mathbf{V}_{n}\left(
\widehat{\mathbf{x}}_{n},\widehat{\mathbf{x}}_{-n}\right)  -\mathbf{F}%
_{n}\left(  \mathbf{V}_{n}|\widehat{\mathbf{x}}\right)  \right \Vert ^{2}.
\label{eq04048}%
\end{equation}

One possibility is to solve (\ref{eq04045}) by exhaustive enumeration. For a
given $K$ we have $g\left(  K\right)  =\frac{\left(  K+1\right)  \left(
K+2\right)  }{2}-3K$ states in which there exist strategy choices and, if we
limit ourselves to pure strategies, for each such state we have $2^{3}=8$
strategy outcome combinations. Hence there exist $h\left(  K\right)
=2^{3g\left(  K\right)  }$ possible overall strategy combinations. In the
following table we list the values of $g\left(  K\right)  $ and $h\left(
K\right)  $ as functions of $K$.

\begin{center}%
\[%
\begin{tabular}
[c]{rrrrrr}\hline
$K$ & $3$ & $4$ & $5$ & $6$ & $7$\\ \hline
$g\left(  K\right)  $ & $1$ & $3$ & $6$ & $10$ & $15$\\
$h\left(  K\right)  $ & $2^{3}=8$ & $2^{9}=\allowbreak512$ & $2^{18}%
=\allowbreak262\,144$ & $2^{30}=\allowbreak1073\,741\,824$ & $2^{45}%
=\allowbreak35\,184\, \allowbreak372\,088\,832$\\ \hline
\end{tabular}
\  \  \  \  \  \
\]
\textbf{Table 1}
\end{center}

\noindent It appears that finding a NE\ by exhaustive enumeration is feasible
\ for , for $K<6$; for larger $K$ values the computational burden is probably unmanageable.

An alternative approach is to use an iterative approach, inspired from
\emph{Value Iteration} which provably yields a solutin for \emph{two-player
zero-sum} games \cite{Filar1997}. Based on this, we can use the
\emph{MultiValue Iteration} (MVI) defined as follows.%
\begin{equation}
\forall n,s:V_{n}^{\left(  t+1\right)  }\left(  s\right)  =F_{n,s}\left(
\mathbf{V}_{n}^{\left(  t\right)  }\left(  s|\mathbf{x}^{\left(  t\right)
}\right)  \right)  . \label{eq04049}%
\end{equation}
If (\ref{eq04049}) converges (which is not a priori guaranteed) it can be
proved that the limit will be a solution of (\ref{eq04045}).

Finally, we can use various optimization algorithms to find a global minimum
(which must be equal to zero) of the function
\begin{equation}
J\left(  \widehat{\mathbf{x}}\right)  =\sum_{n=1}^{3}\left \Vert \mathbf{V}%
_{n}\left(  \widehat{\mathbf{x}}_{n},\widehat{\mathbf{x}}_{-n}\right)
-\mathbf{F}_{n}\left(  \mathbf{V}_{n}|\widehat{\mathbf{x}}\right)  \right \Vert
^{2}. \label{eq04050}%
\end{equation}
We have applied several optimization algorithms, provided by the Matlab
software system, to the minimization of (\ref{eq04050}). Most notably, we have
used the following:

\begin{enumerate}
\item The Matlab function \texttt{fmincon} \cite{MatlabFmincon}, which
performs constrained gradient based optimization \cite{Nocedal2014}.

\item The Matlab function \texttt{particleswarm} \cite{MatlabPSO}, which
implements \emph{particle swarm optimization} \cite{Wang2018}\textbf{.}
\end{enumerate}

We have performed several preliminary experiments and have found out that the
best performance is obtained by the MVI\ algorithm; hence we only report
results obtained by this particular algorithm. The main question is whether
the algorithm converges. We have run 100 repetitions of the following experiment.

After randomly selecting a probabilities vector $\mathbf{p\in}\left(
0,1\right)  ^{3}$ we run MVI\ for 150 runs for each $K\in \left \{
3,4,...,9\right \}  $, always initializing with $\mathbf{x}^{\left(  0\right)
}=\mathbf{0}$. If a solution $\widehat{\mathbf{x}}$ of the optimality
equations (i.e., an equilibrium strategy profile)\ is achieved we count this
as one succesful run of the MVI\ algorithm. After the 150 runs are completed,
we compute, for each $K$, the proportion of succesful runs. The results are
plotted in Figure \ref{fig05011}.

\begin{minipage}{0.45\textwidth}
\medskip
\begin{center}
\includegraphics[scale=0.35]{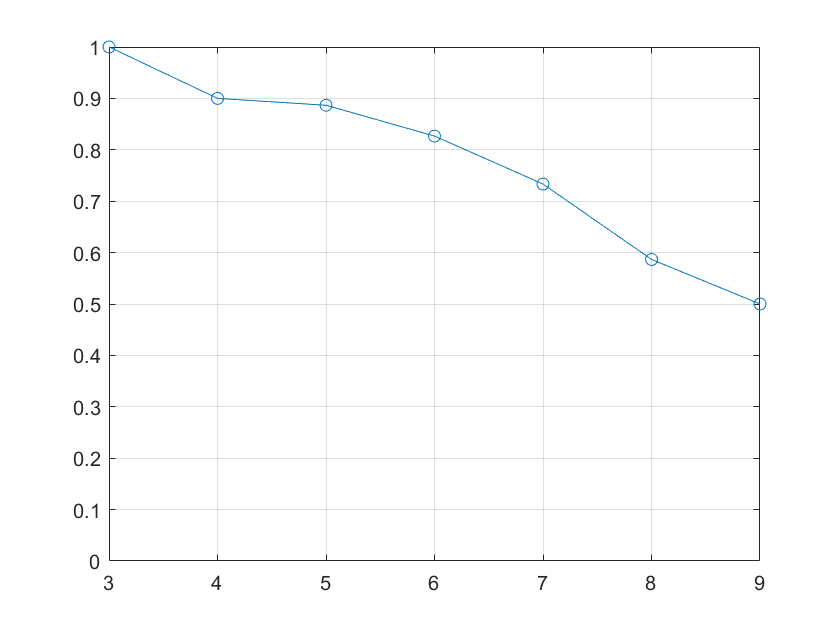}
\captionof{figure}{Proportion of succesful MVI runs for 100 randomly selected $\mathbf{p}$ vectors, as a function of $K$.}
\label{fig05011}
\end{center}
\medskip
\end{minipage}\qquad \begin{minipage}{0.45\textwidth}
\medskip
\begin{center}
\includegraphics[scale=0.35]{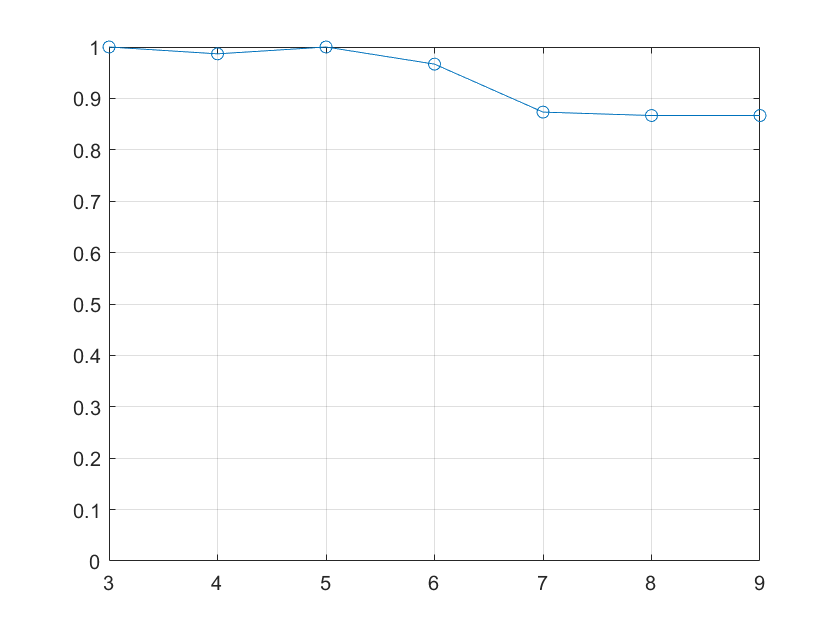}
\captionof{figure}{Proportion of succesful MVI runs for 100 randomly selected $\mathbf{p}$ vectors, as a function of $K$.}
\label{fig05012}
\end{center}
\medskip
\end{minipage}

\noindent It can be seen that MVI\ is \textquotedblleft
relatively\textquotedblright \ succesful. That is, the proportion of convergent
runs is a decreasing function of $K$ (i.e., it is harder to find a solution
for larger $K$ values) and we have a \textquotedblleft
reasonable\textquotedblright \ probability of success for $K\leq7$;
furthermore, even for $K\in \left \{  8,9\right \}  $ we have a 50\% or better
probability of obtaining a solution. It must be emphasized that this is a
lower bound on the performance of MVI. Running MVI\ repeatedly for a specific
$\mathbf{p}$, with random initialization of $\mathbf{x}^{\left(  0\right)  }$,
yields a higher sucess proportion, as our preliminary experiments (not
reported here)\ indicate.

If we look at the MVI\ iteration (\ref{eq04049}) as a dynamical system,
parametrized by $\mathbf{p}$ and $K$, a natural question is this:\ given some
$K$, for which subset $\mathfrak{P}_{K}$ $\subset \left[  0,1\right]  ^{3}$
does $\mathbf{p}\in \mathfrak{P}_{K}$ ensure that (\ref{eq04049})
converges?\ Roughly speaking, the results of Figure \ref{fig05011} indicate
that, for example, the volume of $\mathfrak{P}_{1}$ is 1.00, that of
$\mathfrak{P}_{2}$ around 0.90 and so on. To further investigate this point,
we repeat the above experiment for a restricted set of $\mathbf{p}$ values.
For example, when we fix $p_{1}=1$ and choose $p_{2}$ and $p_{3}$ randomly, we
obtain the results of Figure \ref{fig05012}. Obviously for $\mathbf{p}%
\in \left \{  1\right \}  \times \left[  0,1\right]  ^{2}$ the MVI\ algorithm has
much higher probability of convergence.

We conclude this section by an initial numerical exploration of the following
question:
\[
\emph{how\ much\ does\ a\ player\ benefit\ by\ playing\ optimally?}%
\]
\qquad To answer this, we perform several experiments of the following type.

\begin{enumerate}
\item We fix a particular $\mathbf{p}$ vector and, for $K=9$, use the
MVI\ algorithm to obtain a NE $\left(  \widehat{\mathbf{x}}_{1},\widehat
{\mathbf{x}}_{2},\widehat{\mathbf{x}}_{3}\right)  $ and the corresponding
$V_{1}\left(  \widehat{\mathbf{x}}_{1},\widehat{\mathbf{x}}_{2},\widehat
{\mathbf{x}}_{3}\right)  $.

\item We assign to $P_{1}$ a random strategy $\overline{\mathbf{x}}_{1}$ and
compute his payoff $V_{1}\left(  \overline{\mathbf{x}}_{1},\widehat
{\mathbf{x}}_{2},\widehat{\mathbf{x}}_{3}\right)  $.

\item We compute the difference $\overline{\delta V}_{1}=V_{1}\left(
\widehat{\mathbf{x}}_{1},\widehat{\mathbf{x}}_{2},\widehat{\mathbf{x}}%
_{3}\right)  -V_{1}\left(  \overline{\mathbf{x}}_{1},\widehat{\mathbf{x}}%
_{2},\widehat{\mathbf{x}}_{3}\right)  $.

\item We repeat steps 2 and 3 by assigning the random strategy $\overline
{\mathbf{x}}_{n}$ to $P_{n}$, for $n\in \left \{  2,3\right \}  $.
\end{enumerate}

\noindent For $n\in \left \{  1,2,3\right \}  $, $\overline{\delta V}_{n}$ is a
measure of how much $P_{n}$ benefits by using the equilibrium strategy
$\widehat{\mathbf{x}}_{n}$ rather the random strategy $\overline{\mathbf{x}%
}_{n}$, given that the other two players use the equilibrium strategy. Because
$\left(  \widehat{\mathbf{x}}_{1},\widehat{\mathbf{x}}_{2},\widehat
{\mathbf{x}}_{3}\right)  $ is a NE, we know that $\overline{\delta V}_{n}$
will always be nonnegative; a large $\overline{\delta V}_{n}$ indicates that
$P_{n}$ has a large incentive for using $\widehat{\mathbf{x}}_{n}$.

We also repeat the above suite of experiments comparing the equilibrium
strategy $\widehat{\mathbf{x}}_{n}$ to a uniform strategy $\widetilde
{\mathbf{x}}_{n}$, by which, at every state $\mathbf{s}$, $P_{n}$ selects
equiprobably one of the surviving players; now we compute the quantities
$\widetilde{\delta V}_{n}=V_{1}\left(  \widehat{\mathbf{x}}_{n},\widehat
{\mathbf{x}}_{-n}\right)  -V_{1}\left(  \widetilde{\mathbf{x}}_{n}%
,\widehat{\mathbf{x}}_{-n}\right)  $.

We perform the above experiments for several $\mathbf{p}$ values and present
our results in the following table.

\begin{center}%
\begin{tabular}
[c]{|l|l|l|l|l|l|l|}\hline
&  &  &  &  &  & \\
$\mathbf{p}$ & $\overline{\delta V}_{1}$ & $\overline{\delta V}_{2}$ &
$\overline{\delta V}_{3}$ & $\widetilde{\delta V}_{1}$ & $\widetilde{\delta
V}_{2}$ & $\widetilde{\delta V}_{3}$\\ \hline
$\left(  0.10,0.10,0.10\right)  $ & 0.2357 & 0.2266 & 0.2072 & 0.1995 &
0.1995 & 0.1995\\ \hline
$\left(  0.90,0.00,0.50\right)  $ & 0.0619 & 0.0071 & 0.0726 & 0.0643 &
0.5321 & 0.0863\\ \hline
$\left(  0.80,0.80,0.50\right)  $ & 0.1455 & 0.0339 & 0.0318 & 0.1177 &
0.0462 & 0.0308\\ \hline
$\left(  0.3922,0.8932,0.6634\right)  $ & 0.0088 & 0.1421 & 0.0084 & 0.0091 &
0.1017 & 0.0069\\ \hline
$\left(  0.53,0.20,0.80\right)  $ & 0.0025 & 0.0033 & 0.0948$\times10^{-12}$ &
0.0032 & 0.0026 & 0.0963$\times10^{-12}$\\ \hline
$\left(  0.9000,0.8747,0.2252\right)  $ & 0.0616 & 0.0032 & 0.0006 & 0.0519 &
0.0023 & 0.0009\\ \hline
\end{tabular}

\bigskip

\textbf{Table 1}
\end{center}

\noindent Keeping in mind that the quantities $\overline{\delta V}_{n}$ and
$\widetilde{\delta V}_{n}$ are differences between winning probabilities, the
results of Table 1 indicate that in some cases a player has considerable
incentive to use the equilibrium strategy; this is the case for the
$\mathbf{p}$'s in rows one to four of the table. On the other hand, for the
$\mathbf{p}$'s of rows five and six the player does not gain much by staying
at equilibrium.

\section{Conclusion\label{sec06}}

We have presented a strategic game version of the Three Gamblers' Ruin,
formulated it as a stochastic game and proved that it always has at least one
stationary deterministic NE. We have also briefly investigated the
computational aspects of the game. We believe there is scope for much
additional research on strategic variants of the Gambler's ruin; we conclude
by listing several such variants which we consider worthy of further study in
the future.

\begin{enumerate}
\item The game is played as before but ends when one player is eliminated.
Payoff to each player is his capital at the end of the game.

\item The game is played as previously but the payoff is the \emph{total
\ discounted} capital $\sum_{t}\gamma^{t}s_{n}\left(  t\right)  $.

\item All of the above games can be generalized to games involving $N$ gamblers.

\item The game is played on a graph \ with a gambler on each vertex; gamblers
can only only play against their neighbors.
\end{enumerate}

\end{document}